\documentclass[12pt,letterpaper]{article} 
\pdfoutput=1
 \newcommand{\affaddr}[1]{#1}
 \newcommand{\category}[1]{\noindent\textbf{Category:} #1. }
 \newcommand{\terms}[1]{\newline \textbf{Terms:} #1.\\}
 \newcommand{\keywords}[1]{\textbf{Keywords:} #1.\\}
 \usepackage{amsthm}

\date{}
\usepackage{amssymb} % AMS Fonts
\usepackage{amsmath} % AMSLaTeX

\usepackage{latexsym}
\usepackage[round]{natbib}
\usepackage[svgnames]{xcolor}
\usepackage{comment}
\usepackage{xspace}
\usepackage{enumitem}
\usepackage{hyphenat}
\usepackage[normalem]{ulem}

\usepackage[pdfauthor={Mustafa Elsheikh, Mark Giesbrecht, Andy Novocin,
                       B. David Saunders},
            pdftitle={Fast Computation of Smith Forms of Sparse Matrices
                      Over Local Rings},
            bookmarksnumbered=true,colorlinks,linkcolor=darkblue,
            citecolor=darkgreen, urlcolor=darkred]{hyperref}

\definecolor{darkgreen}{rgb}{0,.35,0}
\definecolor{darkblue}{rgb}{0,0,.5}
\definecolor{darkred}{rgb}{.6,0,0}
\newtheorem{theorem}{Theorem}
\newtheorem{lemma}[theorem]{Lemma}

\newtheorem{corollary}[theorem]{Corollary}
\newtheorem{definition}[theorem]{Definition}

\newtheorem{algorithm}{Algorithm}
\newtheorem{conjecture}[theorem]{Conjecture}

\DeclareMathOperator{\rank}{rank}
\DeclareMathOperator{\diag}{diag}

\DeclareMathOperator{\quo}{quo}
\DeclareMathOperator{\nequiv}{\mskip4mu\not\equiv\mskip4mu}
\DeclareMathOperator{\prob}{prob}

\newcommand{\proj}{\varphi}
\newcommand{\A}{\mathfrak{A}}
\newcommand{\B}{\mathfrak{B}}
\newcommand{\what}{\widehat{w}}

\newcommand{\vhat}{\widehat{v}}
\newcommand{\Ahat}{\widehat{A}}
\newcommand{\Atil}{\widetilde{A}}
\DeclareMathOperator{\M}{\mathsf{M}}
\newcommand{\GR}{{\mathsf{GR}}}

\newcommand{\bigO}{{\mathcal{O}}}
\newcommand{\softO}{{\bigO\mskip1mu\tilde{\,}\mskip1mu}}

\newcommand{\ZZ}{{\mathbb{Z}}}
\newcommand{\QQ}{{\mathbb{Q}}}

\newcommand{\CC}{{\mathbb{C}}}

\renewcommand{\L}{{\mathsf{L}}} 
\newcommand{\F}{{\mathsf{F}}}
\newcommand{\R}{{\mathsf{L}}}

\newcommand{\Fz}{{\F[z]}}

\newcommand{\Fzfe}{{\Fz /(f^e)}}
\newcommand{\xone}[1]{#1} %
\newcommand{\Ln}{{\L^{\xone{n}}}}
\newcommand{\Lmn}{{\L^{m \times n}}}
\newcommand{\Lnk}{{\L^{n \times k}}}
\newcommand{\Lnn}{{\L^{n \times n}}}
\newcommand{\Pden}{{nde}} %

\newcommand{\Fdenden}{{\F^{\Pden \times \Pden}}}

\newcommand{\ZZpe}{{\ZZ_{p^e}}}
\newcommand{\nxn}{{n\times n}}
\newcommand{\lxl}{{\ell\times \ell}}
\newcommand{\calC}{{\mathcal{C}}}
\newcommand{\DD}{{\mathfrak{D}}}
\newcommand{\TT}{{\mathfrak{T}}}
\makeatletter
\@twosidefalse
\makeatother

\begin{document}

\title{Fast Computation of Smith Forms of Sparse Matrices Over Local Rings
\thanks{B. D. Saunders is supported by National Science Foundation Grants CCF-0830130, CCF-1018063. M. Elsheikh, M. Giesbrecht, and A. Novocin are supported by the Natural Sciences and Engineering Research Council of Canada, and MITACS (Canada).}}

\author{Mustafa Elsheikh, Mark Giesbrecht, Andy Novocin\\
    \affaddr{Cheriton School of Computer Science}\\
    \affaddr{University of Waterloo, Canada}
    \and
    B. David Saunders\\
    \affaddr{Department of Computer and Information Sciences}\\
    \affaddr{University of Delaware, USA}
}

\maketitle

\clearpage

\begin{abstract}
We present algorithms to compute the Smith Normal Form of matrices over two families of local rings.  
The algorithms use the \emph{black-box} model which is suitable for sparse and structured matrices.
The algorithms depend on a number of tools, 
such as matrix rank computation over finite fields,
for which the best-known time- and memory-efficient algorithms are probabilistic.

For an $\nxn$ matrix $A$ over the ring $\Fzfe$, where $f^e$ is a power
of an irreducible polynomial $f \in \Fz$ of degree $d$, our algorithm
requires $\bigO(\eta de^2n)$ operations in $\F$, where our black-box is
assumed to require $\bigO(\eta)$ operations in $\F$ to compute a
matrix-vector product by a vector over $\Fzfe$ (and $\eta$ is assumed
greater than $\Pden$).  The algorithm only requires additional storage
for $\bigO(\Pden)$ elements of $\F$.  In particular, if $\eta=\softO(\Pden)$,
then our algorithm requires only $\softO(n^2d^2e^3)$ operations in
$\F$, which is an improvement on known dense methods for small $d$ and
$e$.

For the ring $\ZZ/p^e\ZZ$, where $p$ is a prime, 
we give an algorithm which is time- and memory-efficient when the number 
of nontrivial invariant factors is small.  
We describe a method for dimension reduction while preserving the invariant
factors.
The time complexity is essentially linear in $\mu n r e \log p,$ 
where $\mu$ is the number of operations in $\ZZ/p\ZZ$ to evaluate the
black-box (assumed greater than $n$) and $r$ is the total number of non-zero
invariant factors. 
To avoid the practical cost of conditioning, we give a Monte Carlo certificate,
which at low cost, provides either a high
probability of success or a proof of failure.
The quest for a time- and memory-efficient solution without restrictions on
the number of nontrivial invariant factors remains open. We offer a conjecture
which may contribute toward that end.
\end{abstract}

\bigskip


\clearpage

\section{Introduction}
We consider the problem of computing the Smith Normal Form (SNF) of
sparse matrices over (commutative) local principal ideal rings (PIRs).
The Smith form is a diagonalization of matrices which has many
applications in diophantine analysis~\citep{Chou:1982},
integer programming~\citep{Hu:1969}, combinatorics~\citep{Wallis:1972},
determining the structure of Abelian groups \citep{Newman:1972} and
class groups~\citep{Hafner:1989}, computing Simplicial
Homology~\citep{Dumas:2003}, in system theory~\citep{Kailath:1980,
McMillan:1952}, and in the study of symplectic
spaces~\citep{Chandler:2010}.

The original work of \cite{Smith:1861} proved existence and uniqueness
of the SNF for integer matrices.  The generalization to PIRs is due to
\cite{Kaplansky:1949}.

The problem of computing the Smith form of a sparse matrix over a
principal ideal ring presents several challenges.  One approach is to
simply compute the SNF over the global ring (i.e., $\F[z]$ or $\ZZ$)
and then reduce the result modulo the power of the prime ideal. The algorithm of
\cite{Gie01} for SNF of a sparse matrix over $\ZZ$ could be used, but
the ultimate time requirement is essentially cubic (although space requirements
are lower). An asymptotically faster algorithm along similar lines,
but which requires considerably more space, is presented in~\citep{Eberly:2007}.
The best known algorithm for dense matrices over $\F[z]$
by \citep[Prop~7.16]{Storjohann:2000} requires time essentially equal
to matrix multiplication. However, it is not sensitive to sparsity.

On the other hand, computations over $\F[z]$ and $\ZZ$ suffer from
coefficient growth which is not clearly necessary in a PIR. For
example, over $\ZZ/p^2\ZZ$ where $p$ is a prime, one might hope to perform all
computations modulo $p^2$, and not with integers larger than $p^2.$
\cite{Arne:2003} provides a fast algorithm using elimination in a
PIR, but it is not sensitive to sparsity and requires time proportional
to matrix multiplication.  \cite{Wilkening:2011} demonstrates an
algorithm for dense polynomial matrices over local rings, but offers
no complexity analysis.  \cite{DSV01} give black-box algorithms over $\ZZ$ and
locally at a prime, which however do not have a benefit when only a few
invariant factors are nontrivial.  

When dealing with sparse matrices we would like to preserve the
sparsity of the input matrix, and introduce no fill-in.
Thus we pursue \emph{black-box} algorithms in the sense that the input
matrix is only used for matrix-vector products.
The complexity of black-box algorithms is thus expressed in terms of
the number of matrix-vector products used.  Space requirements are
kept to the storage of a few vectors.
There has been great success in applying black-box methods over finite
and arbitrary fields, starting with \cite{Wiedemann:1986}, where the
cost of many linear algebra problems has been reduced to computing a
linear number of matrix-vector products.  Our ultimate goal is then to
add local Smith form to that list.

Specifically we will consider the local Artinian principal rings (also
known as \emph{special principal rings})
$\ZZ/p^e\ZZ$ for a prime $p$ and positive exponent $e$, and 
$\F[z]/f^e\F[z]$, for irreducible $f \in \F[z]$.
Let $\L$ be a local Artinian principal ideal ring with a maximal prime ideal
$p\L$. For any matrix
$A \in \Lnn$, there exist unimodular matrices $U,V \in \Lnn$ and
a diagonal matrix $S \in \Lnn$ such that $A = USV$,
where
\begin{equation}\label{eqn:SmithDecomposition}
    S =
    \diag(\underbrace{1,\ldots,1}_{r_0},\underbrace{p,\ldots,p}_{r_1},\ldots,
    \underbrace{p^{e-1},\ldots,p^{e-1}}_{r_{e-1}},
    0,\ldots,0)
\end{equation}
\begin{definition}
$S$ is called the {\em Smith form} of $A,$ and the diagonal elements
are called the {\em invariant factors} of $A$.
\end{definition}
Our goal throughout this paper, is to compute the multiplicities of the Smith form
invariants, i.e, $\{r_0,r_1,\ldots,r_{e-1}\}$ for given black-box matrices,
and in particular, sparse matrices. 

\textbf{Our Contribution.}
For matrices over $\Fzfe,$ we present an algorithm which relies on
computing ranks of related black-box matrices over $\F$, and give a
complete complexity analysis.
The key idea of our algorithm is a linear representation of
polynomials in the ring $\Fzfe$ as matrices over $\F,$ and using rank
computations over $\F$ to discover the multiplicities of the Smith invariants.
This reduction allows us to take advantage of the well-studied efficient
algorithms for computing ranks of sparse matrices over fields rather than rings.
The cost of our algorithm is
$\bigO(\eta d e^2 n)$ operations in $\F,$ where each black-box evaluation costs
$\eta$ operations.
Our approach takes a path similar to
the linearization of \cite{Kaltofen:1990} for matrices over
$\F[z]$. This approach is also explored for dense
matrices over local rings by \cite{Wilkening:2011}. 
\citet{Dumas:2009} used rank computations to discover multiplicities
of characteristic polynomial factors for black-box matrices over fields.

The linearization idea, however, would not be applicable over the integers
since there are no appropriate
linear representations from $\ZZ/p^e\ZZ$ to $\ZZ^{\nxn}.$
Hence, it is necessary to develop different methods for the integer case.
A useful approach to Smith form computation is
to begin by determining which primes occur in the invariant factors
and then compute the form locally at those primes.  This has been done
in several recent algorithms \citep{Dumas:2003,Wan:2004}.  A fully
memory efficient, black-box algorithm for sparse and structured
matrices has not been given, however, for lack of an efficient black-box
algorithm for the Smith form locally at a prime.  

Toward that end we give a black-box algorithm over $\ZZ/p^e\ZZ$, whose cost
is essentially dominated by $\mu nek$ for an $n\times n$ sparse matrix
with $\mu$ nonzero entries, $e$ being the largest exponent of the
prime in the Smith form, and $k = \sum_{i=1}^{e-1} r_i$ 
being the number of nontrivial invariant factors ($r_0$ is {\em not} included).
It is to be expected that the cost depends on both $\mu n$ and $e.$
The dependence on $k,$ although unfortunate, is not completely unlikely.
It is natural that it is easier to find Smith form for
matrices with fewer number of non-trivial factors.
In addition, both $e$ and $k$ are small in many cases of interest.  Some
applications with this property are discussed at the beginning of
Section~\ref{Section:modpe}.
The key idea of this algorithm is to apply a reduction in dimension
to dispose
the zero invariant factors, compute a nullspace basis of the reduced matrix to
dispose the ones, and then determine the nontrivial invariant factors
by dense elimination methods. This idea is applicable to the polynomial case
as well. However, it is not interesting, since the complexity of this method 
has an extra factor of $k$ over the rank-based method presented in
Section~\ref{Section:Modfe}.

\textbf{Notation.} Throughout this paper $\F$ denotes a field and $\L$ denotes
a local ring. We use $\ZZ_p$ to denote $\ZZ/p\ZZ$ and
$\ZZ_{p^e}$ to denote $\ZZ/p^e\ZZ.$
 In the complexity analysis, we count the \emph{algebraic
  complexity} in the base field, i.e. we assume that all operations in
the base field have a unit cost. We use ``soft-Oh'' to hide the
logarithmic factors. We say that $f \in \softO(g)$ if there exists a
constant $c$ such that $f \in \bigO(g \log^c g).$ We use $\M(n)$ to
denote the number of operations in the base field required to multiply
two polynomials of degree at most $n$.  Finally, we use
$\bigO(n^\omega)$ to denote the matrix multiplication exponent, e.g.,
$\omega \le 2.372$ using~\cite{Coppersmith:1990}.

\textbf{Roadmap.} The rest of this paper is organized as follows.
Section~\ref{Section:Modfe} contains the
algorithm and analysis for Smith form over $\Fzfe$.
In Section~\ref{Section:modpe}, our
{\em nullspace} algorithm for Smith form over $\ZZpe$ is given after
some discussion of applications, a development of preconditioners for
the problem.  A Monte Carlo verification method is given that can be of use
with small primes.  A conjecture is also given, that may shed some
light on the problem when there are many nontrivial invariants.
Finally, Section~\ref{Section:Conclusion} is a brief summary.

\section{Smith form over {\large$\Fzfe$}}\label{Section:Modfe}
In this section we present an algorithm to compute the local Smith
form of a sparse polynomial matrix.
Throughout this section, let $\F$ be a field, $\L = \Fzfe$
where $f \in \Fz$ is irreducible of degree $d$, and $e \in \ZZ_{>1}$. 
The ideals in this ring are of the form $f^i\L$ for $0 \le i < e$ and the RHS of
equation~\eqref{eqn:SmithDecomposition} becomes 
\begin{equation}\label{eqn:feSmithForm}
   \diag(\underbrace{1,\ldots,1}_{r_0},
             \underbrace{f,\ldots,f}_{r_1}, \ldots,
             \underbrace{f^{e-1},\ldots,f^{e-1}}_{r_{e-1}},0,\ldots,0)
\end{equation}
Our goal is efficiently compute the
multiplicities: $\{r_0,$ $r_1,$ $\ldots,$ $r_{e-1}\}$.
The approach is to embed the ring $\Lnn$ in the ring $\Fdenden$ and
reduce the computation to finding ranks of matrices in the base field, $\F$,
where known fast black-box algorithms can be used.
\subsection{Embedding of {\large $\Lnn$ in $\Fdenden$}}
In this section, we describe the classical embedding of $\L$ into
$\F^{de\times de}$, and how properties of matrices over $\L$ are revealed by
their images over $\F$.  First, define the map $\varphi_e: \L \to \F^{de\times de}$
as follows.
Suppose $f^e=a_0+a_1z+\cdots+a_{de-1}z^{de-1}+z^{de},$ with a companion matrix
\[
C_{f^e}=
\begin{pmatrix}
  0 & 0 & \cdots & -a_0 \\
  1 & \ddots &  & -a_1 \\
  \vdots  & \ddots &  \ddots & \vdots \\
  0 & \cdots & 1 & -a_{de-1}
\end{pmatrix}.
\]
Define $\varphi_e(z)=C_{f^e}$, and $\varphi_e(z^i) = \varphi_e(z)^i$. By linearity,
extend $\varphi_e$ to all elements of $g = g_0+g_1z+\cdots+g_{de-1} z^{de-1} \in
\L$ such that $\varphi_e(g) \in \F^{de \times de}$:
\[
\varphi_e(g) = g(C_{f^e}) = g_0I+g_1C_{f^e}+g_2 C_{f^e}^2+\cdots +g_{de-1} C_{f^e}^{de-1}.
\]
It is straightforward to verify that $\varphi_e$ is a
ring isomorphism between $\L$ and $\F[C_{f^e}]$.

\begin{lemma}\label{lem:PhiRankfi}
  $\rank(\varphi_e(f^i)) = d(e-i)$ for $0\leq i\leq e$.
\end{lemma}
\begin{proof}
  Since $f(C_{f^e})$ acts as multiplication by $f\bmod f^e$, it has
  null vectors which are images of polynomials in $f^{e-1}\L$.  This
  is a vector space of dimension $d$, and hence
  $\rank(f(C_{f^e}))=de-d$.  Also, $f(C_{f^e})$ has minimal polynomial
  $x^{e}$, whence
  \[
  \varphi_e(f) \sim N_f=
  \begin{pmatrix}
    0_d & I_d      & \cdots  & 0_d \\
    \vdots & \ddots &         & \vdots \\
    \vdots & \ddots &  \ddots & I_d \\
    0_d & \cdots & \cdots & 0_d
  \end{pmatrix},
  \]
  where $I_d,0_d\in\F^{d\times d}$ are identity and zero matrices respectively.  
The rank $\varphi_e(f^i)=\varphi_e(f)^i$ is now
evident from the structure of the nilpotent $N_f$.
\end{proof}

We extend the map $\varphi_e$ to $\nxn$ matrices over $\L$. For every
$A \in \Lnn$, $\varphi_e(A)$ is a $\Pden\times \Pden$ matrix over $\F$,
where every entry $a_{i,j}$ of $A$ is replaced by the $de\times de$
block $\varphi_e(a_{i,j})$. Applying $\varphi_e$
to~\eqref{eqn:feSmithForm}, we get
\begin{eqnarray}\label{eqn:PhifeSmithForm}
   \varphi_e(S) = \diag(\underbrace{\varphi_e(1),\ldots,\varphi_e(1)}_{r_0},
             \underbrace{\varphi_e(f),\ldots,\varphi_e(f)}_{r_1}, \ldots, \nonumber \\
             \underbrace{\varphi_e(f^{e-1}),\ldots,\varphi_e(f^{e-1})}_{r_{e-1}},0,\ldots,0) \in \Fdenden.
\end{eqnarray}

\begin{lemma}\label{lem:PhiUnimodular}
  If $U \in \Lnn$ is invertible, then $\varphi_e(U) \in \Fdenden$ is invertible.
\end{lemma}
\begin{proof}
  If $U$ is invertible, then there exists a $W\in \Lnn$ such that $U W
  = I$, and $\varphi_e(U)\varphi_e(W) = \varphi_e(I)$. But
  $\varphi_e(I) = I_{\Pden}$, so $\varphi_e(U)$ has inverse $\varphi_e(W)$.
\end{proof}

We now establish the property relating multiplicities in the invariant factors
of $A\in\L^\nxn$ to the rank of $\varphi_e(\L).$
\begin{theorem}\label{thm:JordanRankProfile}
  Let $A\in\Lnn$ have Smith form in~\eqref{eqn:feSmithForm},
  then $\rank(\varphi_e(A))=der_0+d(e-1)r_1+\cdots +dr_{e-1}$.
\end{theorem}
\begin{proof}
  There exist unimodular matrices $U,V\in\Lnn$ such that $UAV=S$. By
  isomorphism, $\varphi_e(U)\varphi_e(A)\varphi_e(V)=\varphi_e(S)$.  By
  Lemma~\ref{lem:PhiUnimodular}, $\varphi_e(U),\varphi_e(V)$ are invertible and
  thus $\rank(\varphi_e(A))=\rank(\varphi_e(S))$.  In~\eqref{eqn:PhifeSmithForm},
  $\varphi_e(S)$ is a block diagonal matrix, so
  \[ 
  \rank(\varphi_e(S)) = \sum_{i=0}^{e-1}{\rank(\varphi_e(f^i))}. \]
  The proof then follows from Lemma~\ref{lem:PhiRankfi}.
\end{proof}

A consequence of Theorem~\ref{thm:JordanRankProfile} is that, for $1
\le \ell \le e$, we have $\rank(\varphi_\ell(A \bmod f^\ell)) = d\ell
r_0+d(\ell -1)r_1+\cdots +dr_{\ell -1}.$ For example,
$\rank(\varphi_1(A \bmod f)) = dr_0$.  In general, we have the
following corollary. The proof is left to the reader.
\begin{corollary}
Let $\rho_{\ell-1}$ denote $\rank(\varphi_e(A \bmod f^{\ell})),$ $1 \le \ell \le e$. Then
    \begin{equation}\label{eqn:PhiRankLinearSystem}
         \begin{pmatrix}
           d       & 0      & \cdots & 0\\
           2d      & d      & \cdots & 0\\
           \vdots  & \ddots & \ddots & \vdots\\
           ed      & \cdots & 2d     & d
         \end{pmatrix}
         \begin{pmatrix}
           r_0\\ r_1\\ \vdots\\ r_{e-1}
         \end{pmatrix}
         =
         \begin{pmatrix}
           \rho_0\\ \rho_1\\ \vdots\\ \rho_{e-1}
         \end{pmatrix}.
    \end{equation}
\end{corollary}
This system may be solved in linear time.  Let $\sigma_1 = \rho_1$ and
$\sigma_i = \rho_i - \rho_{i-1}$, for $1 < i < e$.  Then the
$\sigma_i$ are the prefix sums of the $r_1$, so $r_1 = \sigma_1$ and
$r_i = \sigma_i - \sigma_{i-1}$, for $1 < i < e$.

Next we consider how to efficiently compute $\{ \rho_0,$ $\rho_1,$
 $\ldots,$ $\rho_{e-1} \}$ for a given black-box matrix.

\subsection{Black-box for the embedding}

Given a black-box for $A \in \Lnn$, over $\L=\F[z]/(f^e)$, we can
easily construct a black-box for $\varphi_\ell(A\bmod
f^{\ell})$, for all $\ell \leq e$, at not much higher cost.

First, we define the black-box model cost model, and then show how to
perform black-box computations under $\varphi_e$ transformations
efficiently.

\begin{definition}\label{def:BlackBox}
  Let $A \in \Lnn$ be a sparse matrix over $\L=\F[z]/(f^e)$, where
  $f\in\F[z]$ is monic and irreducible of degree $d$.  The {\em black-box}
  for $A$ is a mapping $\Ln \to \Ln$ such that for all $v \in
  \Ln$, $Av \in \Ln$ can be computed with $\eta$ operations in $\F$.
  We assume throughout that $\eta\geq \Pden$.
\end{definition}

\begin{lemma}\label{lem:PhiBlackBox}
  Suppose we are given a black-box for $A\in\L^\nxn$, where
  $\L=\F[z]/(f^e)$ as above.  Let $\ell\in\{1,\ldots,e\}$ and
  $\vhat\in\F^{d\ell n}$ with unique pre-image $v\in\F[z]/(f^\ell)$.
  Then we can compute $\varphi_\ell(A\bmod f^\ell)\vhat\in\F^{d\ell
    n}$ with $\bigO(\eta+n\M(de))$ operations in $\F$.
\end{lemma}
\begin{proof}
  Assume that $\vhat \in \F^{d\ell n}$ is labelled as:
  \[
  \vhat=(\vhat_{1,0},\ldots,\vhat_{1,d\ell-1}, \vhat_{2,0},\ldots,\vhat_{2,d\ell-1},
  \ldots, \vhat_{n,0},\ldots,\vhat_{n,d\ell-1}).
  \]
  Construct the vector $v=(v_1,\ldots,v_{n})\in\Ln$, where
  $v_i=\sum_{0\leq j<d\ell} \vhat_{i,j}z^j\in\F[z]$. Now, compute
  $w=Av\bmod f^\ell\in\L^n$ using $\eta$ operations for the black-box
  evaluation plus $\bigO(n\M(de))$ operations in $\F$.  Let $w =
  (w_1,\ldots,w_{n})$. Assume $w_i=\sum_{0\leq
    j<d\ell}\what_{i,j}z^j\in\F[z]$.  Then
  \[
  \what=(\what_{1,0},\ldots,\what_{1,d\ell-1},
  \ldots, \what_{n,0},\ldots,\what_{n,d\ell-1}).
  \]
\end{proof}

Our algorithm for computing the Smith form of a matrix $A\in\Lnn$
given by a black-box is now straightforward. Using
Theorem~\ref{thm:JordanRankProfile} and Lemma~\ref{lem:PhiBlackBox} we
reduce the computation of $\rho_i$'s
in~\eqref{eqn:PhiRankLinearSystem} to computing ranks of matrices over
the ground field $\F$, which can be accomplished using existing fast and
memory-efficient black-box algorithms over fields, e.g. Wiedemann's
algorithm.

Algorithms for computing the rank of a black-box matrix over a field
are developed by \cite{Wiedemann:1986}, and refined in subsequent work
of \cite{Kaltofen:1991}, \cite{Eberly:2004}, and others.  If the input
matrix is in $\F^{n\times n}$ and the black-box evaluation requires $\eta$
operations in $\F$, then the rank algorithms require
$\softO(n \eta)$ operations in $\F$.  They are probabilistic, and
return the correct rank with controllably high probability on any
input.  We will assume that an appropriate choice of black-box rank
method is made, and note that there is considerable difference in
their effectiveness in practice and over various ground fields.

\begin{algorithm}\label{SmithFmodfe}
Smith invariants in $\L=\Fzfe,$ where $f \in \Fz$ is irreducible of degree $d.$\\
  Input: Black-box for $A \in \Lnn.$\\
  Output: $r_0,\dots,r_{e-1}$ such that $r_i$ is the multiplicity of $f^i$ in
          the Smith Form of $A,$ and the multiplicity of $0$ is $n - \sum_{i} r_i.$
\begin{enumerate}[itemsep=-2pt]
\item For all $\ell\in\{1,\ldots,e\}$, invoke a black-box rank algorithm on
  the black-box for $\varphi_\ell(A\bmod f^\ell):\F^{d\ell
    n}\to\F^{d\ell n}$.
  Let $\rho_{\ell-1}=\rank(\varphi_\ell(A\bmod f^\ell))$.
\item Solve~\eqref{eqn:PhiRankLinearSystem} for $r_0,\ldots,r_{e-1}$.
\item Return $r_0, \dots, r_{e-1}.$
\end{enumerate}
\end{algorithm}

\begin{theorem}
  Algorithm~\ref{SmithFmodfe} is correct, and requires $\softO(\eta d
  e^2 n)$ operations in $\F$.  The space requirement of the algorithm
  is $\bigO(d e n)$ elements in $\F$.
\end{theorem}
\begin{proof}
  The correctness follows from the results and discussion in this
  section. We analyze the time and space complexity of step (1), which
  dominates.  This step requires $\bigO(de^2n)$ black-box evaluations, and
  storage for $\bigO(nde)$ elements in $\F$.
\end{proof}

\section{Smith form over {\large$\ZZpe$}}\label{sec:modpe}\label{Section:modpe}
In our experience, most Smith normal forms of integer matrices
in practice involve relatively few non-trivial
factors, i.e., most of the invariant factors are 1's or
0's. The algorithm of this section addresses that situation.

For example, in some work on computing homology of simplicial complexes, \cite{DSV00, DSV01, Babson:1999, Bjorner:1999}, large boundary matrices arose.  
One of the most challenging Smith forms to compute at the time was a 135135 by 270270 matrix which turned out to have 133991 ones, 220 3's, and 
924 zeroes as the invariants.  Other examples in that study were also
large but with even fewer nontrivial invariants.  Most often in
homology computation, the number of 1's (the Betti number) greatly
exceeds the number of non-trivial entries, it seems.

For another example, recently in the study of symplectic 3 spaces,
Smith form computations have been desired of some rather large
matrices \cite{Chandler:2010, Chandler:2011}.  For these matrices it
is conjectured that only one prime will appear in the invariant
factors (other than the largest) and indeed only the local Smith form
at that prime is of interest.  Furthermore, the conjectured structure
predicts only a few nontrivial invariant factors.  We denote these
examples as W3-$q$, where $q$ is a prime power and the Smith form
modulo $q$ is desired.  W3-$q$ is a $\{0,1\}$-matrix of size
approximately $q^3\times q^3$ with about $q^4$ nonzero entries.  The
current challenge is to compute Smith form of W3-64 and W3-81.  The
algorithm described here is designed to handle this case.

For a prime $p$ and an exponent $e \in \ZZ_{> 1}$, let %
$\proj$ be the natural projection $\ZZ_{p^e} \rightarrow \ZZ_p,$
which extends naturally to $\proj: \ZZ_{p^e}^\nxn \rightarrow \ZZ_{p}^\nxn$
by element-wise mapping.  
Note that $x \in \ZZ_{p^e}$ is a unit if and only if $\proj(x) \neq 0$.
Likewise, $A \in \ZZ_{p^e}^\nxn$ is unimodular if and only if $\proj(A)$ is unimodular.

\subsection{Nullspace method}
Let us introduce the approach by way of a sketched example.  
Suppose $A$ is a matrix over $\ZZ_{p^5}$.
Further suppose $A$ is 100 by 100 and 
\[A \sim \diag(1, 1, \ldots, 1,  p, p, p, p^3, p^4, 0, 0, \ldots, 0), \]
with 45 ones and 50 zeroes.
This approximates on a small scale the pattern of invariants we expect to see on W3-$q$.
First, a
reduction in dimension allows us to reduce $A$  
to an $\ell \times \ell$ matrix $\rho(A)$ having the same nonzero invariants,
where $\ell$ is the (max) rank, or slightly larger.  We illustrate with $\ell = 52$:
\[\rho(A) \sim S = \diag(1, 1, \ldots, 1,  p, p, p, p^3, p^4, 0, 0). \]

Over $\ZZ_p$, the nullspace basis $N'$ of $\proj(\rho(A))$ ($N'$ has 7 columns) is
equivalent to that of $S$.  Let $E'$ be the last 7 columns of the $52\times 52$ identity
matrix. Let $E$ and $N$ be arbitrary embeddings of $E'$ and $N'$ in 
$\ZZ_{p^5}^{52\times 7}$, such that $\proj(E) = E', \proj(N) = N'$.  Thus $\rho(A)N$ and
$SE$ are multiples of $p$ and
\[
\rho(A)N \sim SE = \diag(p, p, p, p^3, p^4, 0, 0). 
\]

In summary, the algorithm is to apply a reduction in dimension to dispose of
zeroes, compute nullspace basis $N$ to dispose of ones, and determine the nontrivial invariants by computing Smith form of $AN$ using dense methods.
$AN$ is an $n\times k$ matrix, where $k$ is the number of nontrivial invariants 
(or slightly larger -- $\ell$ - rank mod 2).

Reduction in dimension is a frequent tool and has been used for Smith form, for example, in \cite{DSV01}.  But then their computation proceeds without disposing of the unit invariant factors.  Thus the time complexities below, otherwise similar to theirs, differ in that we replace a rank factor $\ell$ by the number of nontrivial invariants, $k$.
\subsection{Probabilistic dimension reduction}
\label{sec:reduction}

Let $A\in\ZZ_{p^e}^\nxn$, for which we have a fast black-box. Let $A$ have a Smith form
$\diag(s_1,\ldots,s_r,0,\ldots,0)\in\ZZ_{p^e}^\nxn$.  Our goal in this subsection
is to construct $\rho(A)$.  That is, given such an $A$ and a $\ell\in\{1,\ldots,n\}$, to construct a black-box of similar cost for a matrix $B\in\ZZ_{p^e}^\lxl$ which has Smith form $\diag(s_1,\ldots,s_\ell)$,  i.e., with the initial invariant factors of $A$.

Notationally, for integers $n$ and $k\leq n$, let $\calC_k^n$ be the set of
$k$-tuples of distinct elements (in increasing order) of $\{1,\ldots,n\}$.  For a matrix
$B\in\Lnn$, and $\sigma,\tau\in\calC_k^n$, define $B\binom{\sigma}{\tau}$
as the $(\sigma,\tau)$ minor of $B$, i.e., the determinant of the $k \times k$ 
submatrix of $B$ with rows from $\sigma$ and columns from $\tau$.  We use script 
letters, e.g. $\DD, \TT$, to denote matrices with indeterminate entries.

We use techniques similar to that derived in \cite{Gie01} with
\emph{scaled Toeplitz matrix conditioners}. For indeterminates
$\Lambda=\{v_i,w_i,y_i\},$ let $\DD_1 = \diag(v_1,\ldots,v_n),$
$\DD_2 = \diag(w_1,\ldots,w_n),$ and $\TT$ be a generic Toeplitz matrix:
\begin{equation}
\label{eq:precon}
\TT = \begin{pmatrix} 
        y_n & y_{n+1} & \cdots & y_1 \\
        \vdots & y_n & \ddots & \vdots\\
        y_{n-2} & & \ddots & y_{n+1} \\
        y_{2n-1} & y_{2n-2} & \cdots & y_n
        \end{pmatrix}
\end{equation}

\begin{lemma}
  Let $\B=\DD_1\TT\DD_2$ in the indeterminates $\Lambda$, as in
  \eqref{eq:precon}.  Let $k\in\{1,\ldots,n\}$ and
  $\sigma=(\sigma_1,\ldots,\sigma_k)$,
  $\tau=(\tau_1,\ldots,\tau_k)\in\calC_k^n$.
  \begin{itemize}[itemsep=-2pt]
  \item[(i)] $\TT\binom{\sigma}{\tau}\in\ZZ[\Lambda]$ has content $1$;
  \item[(ii)]
    $\B\binom{\sigma}{\tau}=v_{\sigma_1}\cdots v_{\sigma_k} w_{\tau_1}\cdots w_{\tau_k}\TT\binom{\sigma}{\tau}$.
  \end{itemize}
\end{lemma}
\begin{proof}
  Part (i) is from \cite[Lemma 1.3]{Gie01} and part (ii) follows easily from the
  Cauchy-Binet formula.
\end{proof}
Note that $\B\binom{\sigma}{\tau}$ uniquely identifies which minor of $\B$
was selected.

\begin{lemma}
  \label{lem:ddord}
  Let $A\in\ZZ^\nxn$, and $\B_1$, $\B_2$ be $n\times n$ matrices of distinct
  indeterminates from a set $\Lambda$, of the form \eqref{eq:precon}. Then
  $\A=\B_1 A \B_2$ is such that for $1\leq k\leq n$, the content of
  $\psi_k=\A\binom{1\ldots k}{1\ldots k}\in\ZZ[\Lambda]$ equals $\Delta_k$,
  the $k$th determinantal divisor of $A$.
\end{lemma}
\begin{proof}
  By the Cauchy-Binet formula we have
  \[
  \A\binom{1\ldots k}{1\ldots k}
  = \sum_{\sigma,\tau\in\calC_k^n}
  \B_1\binom{1\ldots k}{\sigma} 
  A\binom{\sigma}{\tau}
  \B_2\binom{\tau}{1\ldots k}.
  \]
  Thus $\A\binom{1\ldots k}{1\ldots k}$ is a sum of polynomials of content
  1, with distinct indeterminates, one for each $k\times k$ minor of $A$, times
  the value of that minor.  Hence it must have content equal to the GCD of all
  $k\times k$ minors of $A$, which is equal to the $k$th determinantal
  divisor.
\end{proof}

\begin{theorem}
  \label{thm:condworks}
  Let $A\in\ZZ^\nxn$, $p\geq 6n^2\xi$ a prime, and $\xi\geq 2$.  Let
  $B_1,B_2\in\ZZ^\nxn$ be formed by a random assignment of variables in
  $\B_1,\B_2$ in \eqref{eq:precon} respectively, where choices are made
  uniformly from $L=\{0,\ldots,6n^2\xi - 1\}$, and $\Ahat=B_1AB_2$.  Then with
  probability at least $1-1/\xi$, for \emph{all} $1\leq k\leq n$, the order of
  $p$ in $\Delta_k$, the $k$th determinantal divisor of $A$ 
  equals the order of $p$ in $\Ahat\binom{1\ldots k}{1\ldots k}$.
\end{theorem}
\begin{proof}
  Let $\psi_k$ be as in Lemma \ref{lem:ddord}, which has content equal to the
  $k$th determinantal divisor $\Delta_k$ of $A$.  Observe from our
  construction that $\deg\psi_k\leq 6k\leq 6n$ (the total degree of a $k \times k$ 
  minor of an indeterminate Toeplitz is $ \leq k$).  Thus, with values selected as
  described, by the Schwartz-Zippel Lemma~\citep{Zippel:1979,Schwartz:1980}, we have that
$\psi_k/\Delta_k$ is a polynomial in the entries of matrices $B_1,B_2$ and
  \[
  \prob\left\{ \left( \psi_k/\Delta_k\right)\nequiv 0\bmod p\right\}
  \geq 1-\frac{6n}{6n^2\xi}.
  \]
  This is exactly the probability that the order of $p$ in $\Delta_k$ equals
  the order of $p$ in the leading $k\times k$ minor of $\Ahat$.  The
  probability that this happens for all $k$, from $1\leq k\leq n$ is at least
  $(1-1/(n\xi)))^n\geq 1-1/\xi$.
\end{proof}

\begin{corollary}
  \label{cor:leadingsmith}
  Let $p\geq 6n^2\xi$ be prime, for a $\xi>1$, and $e\geq 1$, and suppose
  $A\in\ZZ_{p^e}^\nxn$ has (local) Smith form
  $\diag(s_1,\ldots,s_n)\in\ZZ_{p^e}^\nxn$.  Let $B_1,B_2\in\ZZ_{p^e}^\nxn$ be
  formed by a random assignments of variables in $\B_1,\B_2\in\ZZ^\nxn$ in
  \eqref{eq:precon} respectively, where choices are made uniformly from
  $L=\{0,\ldots,6n^2\xi - 1\}\bmod p^e$.  Let $\Ahat=B_1AB_2\in\ZZ_{p^e}^\nxn$, and
  for $1\leq k\leq n$ let $\Ahat_k$ be the leading $k\times k$ submatrix of
  $\Ahat$.  Then with probability at least $1-1/\xi$, for \emph{all}
  $k\in\{1,\ldots, n\}$, the local Smith form of $\Ahat_k$ is
  $\diag(s_1,\ldots,s_k)\in\ZZ_{p^e}^{k\times k}$.
\end{corollary}
\begin{proof}
  The Smith form of $A$ equals the Smith form of any $\Atil\in\ZZ^\nxn$
  with $\Atil\equiv A\bmod p^e$, reduced modulo $p^e$ (the only
  non-units will be powers of $p$ after the reduction).  Thus, Theorem
  \ref{thm:condworks} implies that the order of $p$ in the $k$th determinantal
  divisor of $A$ equals the order of $p$ in the leading $k\times k$ minor of
  $\Atil$, for all $k$, with probability at least $1-1/\xi$.  This implies
  that $\Ahat_k=\Atil_k \bmod p^e$ will have Smith form $(s_1,\ldots,s_k)$ for
  all $1\leq k\leq n$ where $\Atil_k$ is the leading $k \times k$ 
  minor of $\Atil$, since $\Delta_k=s_1\cdots s_k$ for $1\leq k\leq n$.
\end{proof}

Computationally, if we know that rank of $A$ is at most $m$, then we can work
with
truncated random scaled Toeplitz matrices $B_1\in\ZZ_{p^e}^{m\times n}$ and
$B_2\in\ZZ^{n\times m}.$ 
Then Corollary \ref{cor:leadingsmith} implies that
$\Ahat=B_1AB_2\in\ZZ_{p^e}^{m\times m}$ 
has the same non-zero invariant factors as $A$.

\textbf{Working with Small Primes.}
The conditions for Corollary \ref{cor:leadingsmith} require $p\geq 6n^2\xi$.
For smaller primes the algorithm may well work, but appears much more
difficult to prove.  The following method may be used 
to remedy this.

The approach is to replace $\ZZ$ by a subring of the ring of algebraic integers in a
number field of degree $\eta=\lceil\log_p(6n^2\xi)\rceil$ over $\QQ$, in which
$p$ is inert.  Specifically, let $\Gamma\in\ZZ[y]$ have degree at least $\eta$ be
such that $\Gamma\bmod p$ is irreducible in $\ZZ_p[y]$.  Let $\gamma\in\CC$ be
a root of $\Gamma$.  Then $\ZZ[\gamma]$ is such that $\ZZ[\gamma]/(p^e)$ is a
local ring which contains $\ZZ_{p^e}$, and such that the residue class field
$\ZZ[\gamma]/(p)$ contains more than $6n^2\xi$ elements.  We call
$\ZZ[\gamma]/(p^e)$ the \emph{Galois ring} with $p^\eta$ elements, and denote it by
$\GR(p^e,\eta)$ (see \cite{McD74}).  Like $\ZZ_{p^e}$, $\GR(p^e,\eta)$ is a
local principal ideal ring with maximal prime ideal generated by $p$.

Analogues of Theorem \ref{thm:condworks} (over $\ZZ[\gamma]$) and Corollary
\ref{cor:leadingsmith} (over $\GR(p^e,\eta)$) can be proven similarly.
We state the latter formally, but leave the proofs to the reader.

\begin{corollary}
  \label{cor:notleadingsmith}
  Let $p$ be prime, $e\geq 1$, $\xi\geq 1$ and
  $\eta=\lceil\log_p(6n^2\xi)\rceil$.  Let $\GR(p^e,\eta)=\ZZ[y]/(\Gamma)$ for
  $\Gamma\in\ZZ[y]$ of degree $\eta$ which is irreducible modulo $p$.  Suppose
  $A\in\ZZ_{p^e}^\nxn$ has (local) Smith form
  $\diag(s_1,\ldots,s_n)\in\ZZ_{p^e}^\nxn$.  Let $B_1,B_2\in\GR(p,\eta)^\nxn$
  be formed by random assignments of indeterminates in $\B_1,\B_2$ in
  \eqref{eq:precon} respectively, where choices are made uniformly from
  $L=\{\sum_{0\leq i<\eta}\alpha_iy^i: \alpha_i\in\{0,\ldots,p-1\}\}\bmod p^e$.
  Let $\Ahat=B_1AB_2\in\GR(p,\eta)^\nxn$, and for $1\leq k\leq n$ let
  $\Ahat_k$ be the leading $k\times k$ submatrix of $\Ahat$.  Then with
  probability at least $1-1/\xi$, for \emph{all} $k\in\{1,\ldots, n\}$, the
  local Smith form of $\Ahat_k$ is $\diag(s_1,\ldots,s_k)\in\ZZ_{p^e}^{k\times
    k}$.
\end{corollary}

\subsection{Probabilistic validation of dimension reduction}
This section might be titled, ``Escaping the tyranny of Schwartz-Zippel."  
It is frequently the case, as is exemplified in the previous section, 
that an argument for a favourable probability based on the Schwartz-Zippel Lemma 
requires an inconveniently large set for random assignments.  
Experience in practice demonstrates that far smaller sets suffice virtually always.  
In this section we give a method to have a provably low probability of failure 
while making no assumptions about the basic preconditioner.

For a matrix $A$ let $s_k(A)$ denote the $k$-th invariant factor of $A$ (in the order in which $s_i(A) \mid s_{i+1}(A)$, for $1 \leq i < n$). 
Let $S_k(A)$ denote 
the leading $k\times k$ submatrix of the Smith form of $A$, 
$S_k(A) = \diag(s_1(A), s_2(A),\ldots, s_k(A))$.
Let $[A, B]$ denote the side by side join of two conformable matrices. %

\begin{theorem}\label{oneside}
Let $A \in \Lmn$, $Q \in \Lnk$, and $y \in \L^{n}$, with $k \leq \min(m, n)$ and the entries of $y$ being uniform random variables over $\L$. 
Then $AQ$ has the first $k$ invariant factors of $A$ with high probability if $AQ$ and $A[Q,y]$ have the same first $k$ invariant factors.  To be precise, we have the following conditional probability:
\[
\prob(S_k(A) = S_k(AQ) ~|~ S_k(AQ) = S_k([AQ , Ay])) \geq 1 - 1/p.
\]
\end{theorem}

\begin{proof}
Note that $S_k(A) = S_k(B)$ if and only if the $i$-th determinantal divisors are equal: 
$D_i(A) = D_i(B), i \in 1, \ldots, k$, where $D_i(A)$ denotes the GCD of all 
$A\binom{\sigma}{\tau}, \sigma, \tau \in \calC_i^n$. 
We will call any $i\times i$ minor equal to $D_i(A)$ (up to unit multiple) a {\em witness} 
for the determinantal divisor $D_i(A)$.  
Note that the GCD of a set of elements of $\L$ is a member of the set (up to unit multiple), 
so every determinantal divisor has at least one witness.  
For example, for a Smith form $S$, the $j$-th determinantal divisor, if not a unit, 
has exactly one witness, $D_j(S) = S\binom{\iota(j)}{\iota(j)}$, 
where we define $\iota$ by $\iota(j) = (1, 2, \ldots, j)$.

Without loss of generality, suppose for notational symplicity that $k \leq n \leq m$.
Let $A = USV$ be the Smith form factorization of $A$ with 
$S = \diag(s_1, \ldots, s_n)$ and $U,V$ unimodular.
We are concerned with the invariants of $A, AQ, $ and $A[Q, y]$.  
Multiplication by unimodular $U^{-1}$ does not affect invariants, 
so let us consider $SV, SVQ,$ and $S[VQ,Vy]$.  
Because $V$ is unimodular, $Vy$ is a uniform random variable if $y$ is.  
Also, we have made no conditions on $Q$, so we simplify notation, 
substituting $Q$ for $VQ$ and $y$ for $Vy$. 
In other words, without loss of generality we may assume $A = S$ is in Smith form.

We will show the contrapositive of our proposition.  
Suppose $j$ is the first index at which $S_j(A) \neq S_j(AQ)$.
Then $A\binom{\iota(j)}{\iota(j)} \neq A\binom{\iota(j)}{\iota(j)} Q\binom{\iota(j)}{\sigma}$, for all $\sigma \in \calC_j^k$.  Otherwise such a minor would witness equality of $S_j(A)$ and $S_j(AQ)$.
It follows that $p \mid Q\binom{\iota(j)}{\sigma}$. 

Because $j$ is the first such case, there must be an $(j-1) \times (j-1)$ minor $Q\binom{\iota(j-1)}{\tau}$ which is a unit (not divisible by $p$),
where $\iota(j-1), \tau \in \calC_{j-1}^k$.
Let $\tau'$ denote $\tau \cup \{k+1\}$, a column index set for $[Q,y]$ which includes the last column.
Then the expansion of the $\iota$ by $\tau'$ minor of $A[Q, y]$ has as coefficient of $y_j$ the term
$A\binom{\iota(j)}{\iota(j)}Q\binom{\iota(j-1)}{\tau}$, which is a unit.  Thus, for each setting of $y_1, \ldots, y_{j-1}$, there is at most one value modulo $p$ or $y_j$ for which $p \mid [Q,y]\binom{\iota(j)}{\tau}$.  For all other values, this minor witnesses that $D_j(AQ) \neq D_j(A[Q,y]$.  Thus 
when $S_k(A) \neq S_k(AQ)$ there is at most a $1/p$ chance that $S_k(A[Q,y])$ agrees with $S_k(AQ)$.
\end{proof}

The following corollary allows us to verify or disprove the success of preconditioners such as those used in dimension reduction.  This works when the theory of the preconditioner's probability of success is invalidated because of an insufficiently large set from which random values are chosen.
Thus in practice for a small prime, one can skip the domain extension described at the end of Section~\ref{sec:reduction}.  Indeed, no theory of the preconditioner is required at all.  One can try to ``get lucky'' and skip preconditoners altogether.  For instance, apply the corollary where $PAQ$ selects the leading $k\times k$ submatrix of $A$.  If the method validates, then the Smith form was found economically, otherwise try a more thorough preconditioning.

\begin{corollary}
[$\mbox{projection verification}$]
Let matrix $A \in \ZZ^\nxn$ be given and heuristic preconditioners $P \in \ZZ^{k\times n}$ and $Q \in \ZZ^{n\times k}$.
Choose $R_1 \in \ZZ^{n\times c}, R_2 \in \ZZ^{c \times n}$ at random.
Let $B = \left(\begin{matrix}
PAQ &  PAR_1\\
R_2AQ & R_2AR_1 \\
\end{matrix}\right) \in \ZZ^{k+c \times k+c}.
$
If $S_k(PAQ)$ = $S_k(B)$ then these are the first $k$ invariant factors of $A$ with probability greater than $1 - 2/p^c$.
\end{corollary}
\begin{proof}
Let $C = \left( \begin{matrix} PA\\R_2A \end{matrix}\right)$ so that $B = [CQ, CR_1].$
Note that the condition $S_k(PAQ)$ = $S_k(B)$ implies that all minors of $B$ containing $PAQ$ have those invariants.  In particular $B$ and its left side $CQ$ have the same invariants.
So Theorem \ref{oneside} applies $c$ times for each of the columns of $R_1$.  Because these are independent random vectors the probability that $S_k(CQ) = S_k([CQ,CR_1])$ when $S_k(CQ) \neq S_k(C)$ is at most $1/p^c$.
  Then apply Theorem \ref{oneside} to $A, PA, C$ on the left $c$ times for the $c$ rows of $R_2$.  
Again, the probability of an unfortunate equality is at most $1/p^c$, and otherwise $S_k(PA) = S_k(A)$ is validated.  Thus the overall probability of success is at least $(1- 1/p^c)^2 > 1 - 2/p^c$.
\end{proof}

The idea to use some random dense rows in preconditioning is widespread.  
It was used already in \citep[proof of theorem 1]{Wiedemann:1986}, 
and in a sense it is the basis of block iterative methods.  
The idea to obtain a good probability (especially for small primes) by solving a sparse problem twice, 
first without the few random dense rows and/or columns then with them, was used in \cite{SY09}. 
For integer lattices, there are also some similarities to the additive preconditioners of
\cite{EbeGie00}, and the lattice compression of \cite{CheSto05},
especially in the analysis for small primes dividing invariant factors.

The parameter $c$ can be adjusted to get the desired degree of certainty.  For example, $c =21$ ensures probability of failure less than one in a million, since $2/p^c \leq 1/2^{20} \leq 10^{-6}$ in that case.  Also, one can pad with more random rows and columns to improve weak preconditioners.  Thus if $c = 30$ is used and the matrix with the first 9 of those columns and rows has the same Smith form as the matrix with all 30 random rows and 30 random columns adjoined, we have computed the Smith form with error expected less than once in a million trials.  In effect, we have corrected for some weakness in the heuristic preconditioners with 9 extra rows and columns and then verified with 21 more.

\subsection{Algorithm for the Smith Normal Form}

After reducing the dimension to a value at or near the number of nonzero invariant factors, the following algorithm is applied.
Recall that $\proj$ is the natural projection $\ZZ_{p^e} \to \ZZ_p$.
\begin{algorithm}\label{Smithpe-nullspace}
Smithpe-nullspace\\
Input: a black-box for $B \in \ZZ_{p^e}^{n \times n}$, and
a bound $\ell$ for the number of nonzero invariant factors\\
  Output: $S$, the Smith form of $B$.
  \begin{enumerate}[itemsep=-2pt]
\setcounter{enumi}{-1}
  \item 
Set $A$ to the 
dimension reduction of $B$ to $\ell\times \ell$. 
  \item 
Let $r_0 = \rank(\proj(A))$ over $\ZZ_p$. The nullity of $\proj(A)$ is then $k = \ell-r_0$. 
  \item
Compute $N' \in \ZZ_{p^e}^{\ell \times k}$, a lifting to $\ZZ_{p^e}$ of a right nullspace basis of $\proj(A)$ over $\ZZ_p$. 
  \item Let $N = AN' \in \ZZ_{p^e}^{\ell \times k}$. This involves $k$ matrix vector products with $A$.
Note that $N$ is divisible by $p$.
  \item
Compute the Smith normal form of $N$ over $\ZZ_{p^e}$ by Gaussian elimination:
\[ 
\diag(\underbrace{p,\ldots,p}_{r_1},\underbrace{p^2,\ldots,p^2}_{r_2},\ldots,
\underbrace{p^{e-1},\ldots,p^{e-1}}_{r_{e-1}},
\underbrace{0,\ldots,0}_{r_{e}}).
\]
\item Return $r_0, \dots, r_{e-1}.$
  \end{enumerate}
\end{algorithm}

We will analyze the algorithm holding $e$ and $p$ constant.  Considering them as parameters would introduce a factor of $\softO(e\log(p))$.
Let the cost of matrix-vector product by $B$ is $\bigO(\mu)$.  
Since we are holding $e$ and $p$ constant, this is the same for application to vectors in $\ZZ_p^n$ and in $\ZZ_{p^e}^n$.

Step 0: Toeplitz matrices may be applied to vectors via polynomial multiplication, so the cost of the black-box for $A$ is $\bigO(M(n) + \mu)$.  
$M(n)$ is $\softO(n)$ and we will assume $\mu \geq n$, so the black-box cost of matrix-vector product by $A$ is $\softO(\mu)$.

Step 1: 
The rank over $\ZZ_p$ can be done by a black-box method in $\softO((\ell\mu)\log(\xi))$ to achieve probability of error less than $1/\xi$ \citep{Wiedemann:1986}.
Memory requirement is $\bigO(1)$ vectors in $\ZZ_p^\ell$.

Step 2: 
Let $k = \ell - r_0$ denote the nullity of $A$ modulo $p$.
By black-box methods, $k$ random samples of the nullspace will yield a nullspace basis $N'$. Oversampling can be done and column echelon form computation used to reduce to a basis of $k$ columns if need be.  The cost is $\softO(k(\ell\mu))$.  Space is $\bigO(k\ell)$.  For instance see~\cite{Chen:2002}.

Step 3: The cost of applying $A$ to $N'$ is $\softO(k\mu)$.

Step 4: Any nullspace for $S$ over $\ZZ_p$  is of the form $EW'$,
where $E$ is the last $\ell-r_0$ columns of the identity and $W'$ is $k\times k$ unimodular.  
Then $0 = AN' = USVN = USEW'$ modulo $p$, 
for some unimodular $W'$.  This lifts to a factorization $AN = USEW$ modulo $p^e$ with $U, W$ unimodular.
Thus $AN$ has the Smith form $SE$.
The local Smith form of this dense matrix can be computed by elimination.  
An algorithm running in $\softO(\ell k^{\omega-1})$ is in \cite{Storjohann:2000}. 

\begin{theorem}
Algorithm~\ref{Smithpe-nullspace} is a correct Monte Carlo algorithm for computing Smith normal form
over $\ZZ_{p^e}.$
The time complexity is $\softO(\ell k (k^{\omega - 2} + \mu))$,  
where $k$ is the number of nontrivial (neither 0 nor 1) invariant factors, $\ell$ 
is the reduced dimension (which can be the rank), 
and $\mu$ is the cost of matrix vector product by $B$.
Note that the time complexity is $\softO(\ell k \mu)$ under the very modest assumption that $k^{\omega-2} < \mu$.
The memory requirement is $\bigO(k\ell)$.
\end{theorem}

\subsection{A conjecture about {\large $p$}-adic carries}

While attempting a general approach for the $\ZZpe$ case without using
dense elimination and while using only rank computations over the
field $\ZZ_p$, we discovered an interesting pattern about ranks of matrices
over $\ZZpe$ which could be of independent interest.
To put the discovered conjecture in a proper context, we first introduce the
attempt which lead to discovering it, then we present the conjecture
which discourages this approach.

Consider an element $\alpha \in \ZZpe$ written in terms
of its \emph{unique $p$-adic expansion} as $\alpha = \alpha_0  + \alpha_1 p + \ldots + \alpha_{e-1}
p^{e-1}$ where $\alpha_i \in \ZZ_p$ for $0 \le i \le e-1.$ This representation
extends naturally to matrices over $\ZZpe$, i.e.,
$A = A_0 + p A_1 + \ldots + p^{e-1} A_{e-1},$ where $A_i \in \ZZ_p^{n \times n}.$
For clarity of presentation and limited space, we only confine the
discussion to the case $e = 2.$ Expanding $A = USV,$
we get $(A_0 + p A_1) = (U_0 + p U_1)(S_0 + pS_1)(V_0 + pV_1),$ 
where $A_0  =  U_0S_0V_0 \bmod p$ and
\begin{eqnarray}
 A_1 & = & U_1S_0V_0 + U_0S_0V_1 + U_0S_1V_0 \label{Eq:invalid}
\end{eqnarray}
where $\rank(S_0) = r_0,$ $\rank(S_1) = r_1,$ which are the multiplicities of
$1$'s and $p$'s in the Smith form diagonal, respectively.
Furthermore, with appropriate preconditioning\footnote{Details are outside
the scope of presenting this conjecture.},
$\rank(A_0)$ is proportional to $r_0$, and $\rank(A_1)$ is proportional to
$r_0 + 2r_1.$ Hence, this formulation leads to a belief that we can easily
\emph{isolate} $A_0, A_1,$ compute their ranks over $\ZZ_p$, and discover 
multiplicities of the invariant factors.
However, by closer inspection, equation~\ref{Eq:invalid} is in fact:
\[ A_1 = U_1S_0V_0 + U_0S_0V_1 + U_0S_1V_0 + \overbrace{U_0S_0V_0 \quo p}^{\text{carry}}, \]
where the extra term, $(U_0S_0V_0 \quo p),$ is introduced by the fact that operations
in $\ZZpe$ exhibit carries. As a simple example $5 \cdot 5$ over $\ZZ_{3^4}$
is $(2 + p)(2 + p) = 1 + 2 p + 2 p^2,$ where $2 p^2$ is a carry term.
These carries contribute to the overall ranks of matrix expressions, and present
a challenge in computing the Smith form. We hoped to reasonably 
bound the ranks of these carries, so a working algorithm could be developed which
reduces Smith form computation to efficient rank computations over $\ZZ_p$.
This approach would be superior to algorithms presented in previous section and
literature since it preserves sparsity.
It is worth noting that in the polynomial case, this approach
yields a working algorithm. The following conjecture illustrates why the 
$p$-adic case is more difficult to resolve in this way.

\begin{conjecture}\label{RankCarryConjecture}
Assume $p$ is a prime, $U,S,V \in \ZZ_{p}^{n \times n}$ such that $U,V$ are
 invertible, $S = \diag(1,\ldots,1,0,\ldots,0)$ and $\rank(S) = r$.
 Let $M = USV = M_0 + M_1 p + \ldots + M_s p^s,$
 where $s = \bigO(\log_p{r})$ and $M_i \in \ZZ_{p}^{n \times n}.$
 We conjecture that when $p = 2,$ \[ \rank(M_i) \leq \binom{r}{2^i}, \]
 and when $p = 2k + 1,$ we conjecture that
  \[ \rank(M_1) \leq
         \sum_{i=0}^{k}{ \binom{r+2i}{2i+1} } + \binom{r+2k - 1}{2k} -2r.
  \]
  Furthermore, in the generic case where
  $U,V$ are uniformly chosen at random in $\ZZ_p$, and $n$ is 
  arbitrarily large, the ranks are \emph{equal} to bounds above.
\end{conjecture}
This conjecture shows that a large dimension product of matrices with entries in 
$\ZZ_p$ of small rank 
can still have very large, but not full, rank carry matrices. These carries 
will impact many digits in the expanded product.
The evidence for the conjecture is experimental\footnote{Code is available
at: \url{http://cs.uwaterloo.ca/~mwg/smithpe/}}.
We developed the formulas above by reverse engineering sequences of computed
ranks resulting from experiments with different primes and matrices of different
dimensions.

In the generic case of equality,
the conjecture could be used to generate an algorithm for small primes, e.g.
when $p = 2.$
However, without further refinements, this approach yields 
an exponential time algorithm which is prohibitive.

\section{Conclusion}\label{Section:Conclusion}

We have given efficient algorithms for computing Smith Normal Form over two local rings: $\Fzfe$ and $\ZZ/p^e\ZZ$.  These are useful components for SNF algorithms over $\F[z]$ and $\ZZ,$ respectively.  These algorithms are efficient in the black-box model, which means that they are well suited to sparse and structured matrices.
The integer algorithm is output sensitive.  Its memory and time usage grows in proportion to the number of nontrivial invariant factors.  A memory efficient algorithm without that restriction has not been found.  In addition, we gave a conjecture about the rank of carries resulting from multiplying single-digit 
$p$-adic matrices.

\def\bibfont{\small}
\def\bibsep{4pt}

\def\newblock{\hskip .11em plus .33em minus .07em}

\end{document}